\documentclass[10pt]{llncs}

\usepackage{amssymb}
\usepackage{amsmath}
\usepackage{epsfig}
\usepackage{graphics}
\usepackage{color}
\usepackage{llncsdoc}
\spnewtheorem{observation}{Observation}{\bfseries}{\itshape}

\begin{document}

\title{New Algorithms for Minimizing the Weighted Number of Tardy Jobs On a Single Machine\thanks{This research was supported by Grant No. 2016049 from the United States-Israel Binational Science Foundation. The first author is also supported by the People Programme (Marie Curie Actions) of the European Union’s Seventh Framework Programme (FP7/2007-2013) under REA grant agreement number 631163.11, and by the Israel Science Foundation (grant no. 551145/14).}}

\author{%
Danny Hermelin \inst{1}
\and%
Shlomo Karhi \inst{2}
\and%
Mike Pinedo \inst{3}
\and%
Dvir Shabtay \inst{1}}

\institute{
Department of Industrial Engineering and Management, Ben-Gurion University, Israel \\
\email{hermelin@bgu.ac.il, dvirs@bgu.ac.il}
\and
Department of Management, Bar-Ilan University, Israel\\
\email{shlomo.karhi@biu.ac.il}
\and
Stern School of Business, New York University, USA.\\
\email{mpinedo@stern.nyu.edu}
}

\date{}

\maketitle

\begin{abstract}

In this paper we study the classical single machine scheduling problem where the objective is to minimize the weighted number of tardy jobs. Our analysis focuses on the case where one or more of three natural parameters is either constant or is taken as a parameter in the sense of parameterized complexity. These three parameters are the number of different due dates, processing times, and weights in our set of input jobs. We show that the problem belongs to the class of fixed parameter tractable (FPT) problems when combining any two of these three parameters. We also show that the problem is polynomial-time solvable when the latter two parameters are constant, complementing Karp's result who showed that the problem is NP-hard already for a single due date.

\end{abstract}

\section{Introduction}

In this paper, we analyze the tractability of the NP-hard problem of minimizing the weighted number of tardy jobs on a single machine when one or more of three natural parameters is either constant or is taken as a parameter in the sense of parameterized complexity. This problem is formalized as follows: We are given a set of $n$ jobs $\mathcal{J}=\{J_{1},...,J_{n}\}$ to be scheduled non preemptively on a single machine. Associated with each job $J_{j} \in \mathcal{J}$, are three non-negative integers $p_{j}$, $d_{j}$ and $w_{j}$, which represent the processing time, due date and weight of $J_j$, respectively. A schedule (or a solution) to our problem is defined by an ordering, $\pi:\mathcal{F} \to \{1,\ldots,n\}$, of the jobs in $\mathcal{J}$ representing the sequence in which the jobs are to be processed on the single machine. For a given schedule $\pi$ of $\mathcal{J}$, 
the set of \emph{tardy jobs} is the set of all jobs $J_j$ whose completion time $\sum_{\pi(J_i) \leq \pi(J_j)} p_i$ is greater than their due date $d_j$ (the complementary set is commonly referred to as the set of \emph{early jobs}). Our objective is to compute a schedule that minimizes the weighted number of tardy jobs. Following the standard three field notations introduced by Graham \emph{\textit{et al.} }\cite{Graham79}, we refer to our problem as the $1\left\vert {}\right\vert \Sigma w_{j}U_{j}$ problem, where $U_j$ can be thought of as a binary indicator variable that is equal to 1 iff $J_j$ is tardy in a corresponding schedule.

The $1\left\vert {}\right\vert \Sigma w_{j}U_{j}$ problem is a fundamental problem in the field of combinatorial optimization in general and in scheduling theory in particular. The problem was studied already in the late 60s~\cite{lawler1969functional}, and perhaps even before that. Karp placed the problem in the pantheon of combinatorial optimization problems by listing it in his landmark paper from 1972~\cite{Kar72}. There it is shown that the problem is NP-hard even if all due dates are equal, giving the first NP-hardness proof for any scheduling problem. The reduction presented in~\cite{Kar72} is from the 0-1~Knapsack problem, and in fact, it is not difficult to show that this reduction works in both ways. 
Thus, the $1\left\vert {}\right\vert \Sigma w_{j}U_{j}$ problem is a generalization of the 0-1~Knapsack problem where jobs may have arbitrary due dates.

Besides being the first scheduling problem which was shown to be NP-hard, the $1\left\vert {}\right\vert \Sigma w_{j}U_{j}$ problem also has a focal role in the history of algorithm development. The classical algorithm of Lawler and Moore~\cite{lawler1969functional} is one of the earliest and most prominent examples of dynamic programming, and of a pseudo-polynomial time algorithm. Sahni~\cite{sahni1976algorithms} used the $1\left\vert {}\right\vert \Sigma w_{j}U_{j}$ problem as one of the three first examples to illustrate the important concept of fully polynomial time approximation schemes (FPTAS's) in the area of scheduling. To that effect, several generalizations of the $1\left\vert {}\right\vert \Sigma w_{j}U_{j}$ problem have been studied in the literature, testing the limits to which these techniques can be applied~\cite{Adamu2014}.

Despite all this, there are not many papers that provide exact algorithms for solving the $1\left\vert {}\right\vert \Sigma w_{j}U_{j}$ problem. As mentioned above, Lawler and Moore~\cite{lawler1969functional} and Sahni~\cite{sahni1976algorithms} provided dynamic programming procedures to solve the problem in pseudo-polynomial time, showing that the problem is only weakly NP-hard. In fact, their algorithms run in polynomial-time if either all weights, all processing times, or all due dates are integer values that are bounded by a polynomial function in $n$. Exact algorithms based on a Branch-and-Bound procedure have been presented by Villarreal and Bulfin~\cite{villarreal1983scheduling}, Tang~\cite{tang1990new} and M'Hallah and Bulfin~\cite{m2003minimizing}. Moreover, the problem is known to be polynomial-time solvable in a few special cases: Moore~\cite{Moore1968} provided an $O(n\log n)$ time algorithm for solving the unit weight $1\left\vert {}\right\vert \Sigma U_{j}$ problem, and Peha~\cite{Peha1995} presented an $O(n\log n)$ time algorithm for the case where all jobs have equal processing time (see also~\cite{BruckerKravchenko2006}).

In the next subsection, we present the basic concepts of parameterized complexity theory (for more details we refer the reader to Cygan \emph{et al.}~\cite{Cygan2015}, Downey and Fellows~\cite{DF99}, Flum and Grohe~\cite{FG98}, and Niedermeier~\cite{N06}). This will enable us to clearly present our research goals, in the subsection that follows. We conclude the introduction section by citing some related work, and by providing a roadmap for the rest of the paper including details about the techniques that will be used.

\subsection{Basic concepts in parameterized complexity theory}

The main objective in parameterized complexity theory is to analyze the tractability of NP-hard problems with respect to other natural
problem parameters, and not only with respect to their input length. For this, problem instances are ordered pairs of the form $(x,k) \in \{0,1\}^* \times \mathbb{N}$, where $x$ is a binary string that denotes the actual input, and $k$ is a numerical value that quantifies the parameter (or set of parameters). The following definition is the central notion of parameterized complexity theory.

\begin{definition}
A problem $\Pi \subseteq \{0,1\}^* \times \mathbb{N}$ is fixed-parameter tractable and belongs to the complexity class (FPT) if there is an algorithm that can determine whether any instance $(x,k) \in \{0,1\}^* \times \mathbb{N}$ is in $\Pi$ in $f(k) \cdot |x|^{O(1)}$ time, where $f$ is some computable function that solely depends on $k$.
\end{definition}

The main issue here is to differentiate between those problems that require $f(k) \cdot |x|^{O(1)}$ time, and those that require $|x|^{f(k)}$ time. To exemplify this distinction, consider two classical graph theoretic problems Independent Set and Vertex Cover (see \emph{e.g.}~\cite{GJ79} for formal definitions), both parameterized by the cardinality of the solution. While there are algorithms for Vertex Cover that run in $O(2^k n)$ time, and even better~\cite{N06}, we do not know any algorithm for Independent Set that runs in $n^{o(k)}$ time. In fact, parameterized complexity theory provides convincing evidence that no such algorithm exists. The reader should note that due to this, we are able to solve Vertex Cover instances with much larger solution sizes in comparison to Independent Set instances; in particular, instances with solution size $O(\lg n)$ can be solved in polynomial time.

\subsection{Parameterized tractability of scheduling problems}

Our paper can be considered as another attempt to facilitate the tools of parameterized complexity into the area of scheduling. Despite the rich amounts of interesting NP-hard problems the latter area has, and the many successes of the former in designing tractable algorithms for NP-hard problems, there has been disappointingly little research in combining the two worlds.

Two early papers by Bodlaender and Fellows~\cite{BodlaenderFellows1995}, and Fellows and McCartin~\cite{FellowsandMcCartin2003} studied scheduling problems with precedence constrains. Both these papers obtain only hardness results for their respective problems under consideration. The first to provide positive results for scheduling problems in the perspective of parameterized complexity are Mnich and Wiese \cite{MnichandWiese2013} who showed that various classical scheduling problems on parallel machines, and on a single machine with rejection are fixed-parameter tractable with respect to several natural parameters. Van Bevern \emph{et al.}~\cite{journals/corr/BevernNS15} study a fixed interval scheduling problem where jobs are to be scheduled on a set of identical machines working in parallel, and show FPT results with respect to several combinations of interesting structural parameters of the problem. Hermelin \emph{et al.}~\cite{hermelin2015scheduling} study several single machine two agent scheduling problem, when the number of jobs belonging to one of the agents is taken as a parameter.

Additional papers that study the parameterized tractability of scheduling problems are those by van Bevern \emph{et al.}~\cite{vanBevern2016a,vanBevern2016b}, Cieliebak \emph{et al.}~\cite{Cieliebak2004}, and Knop \emph{et al.}~\cite{DBLP:journals/corr/KnopK16}.

\subsection{Our contribution}

In this paper we present new exact algorithms for the $1\left\vert {}\right\vert \Sigma w_{j}U_{j}$ problem. Since the problem is already known to be pseudo-polynomial time solvable~\cite{lawler1969functional,sahni1976algorithms}, we focus on the case where the input can contain arbitrary large (\emph{i.e.}, exponential) integer values. Moreover, our interest will focus on the following three parameters in our input instance:

\begin{itemize}
\item $nu_{d} =$ The number of different due dates.
\item $nu_{p} =$ The number of different processing times.
\item $nu_{w} =$ The number of different weights.
\end{itemize}

An example scenario where the first parameter is relevant is when delivery costs are high and thus products are batched to only few shipments. In such a case, each job may be assigned a due date by the marketing department according to one of the planned delivery dates. An example for the second parameter is when the number of job types that a manufacturer produces is limited, though each job might have different importance and a different due date. The last parameter corresponds to the case where customers are batched into few subsets according to their importance, and all customers within the same subset are similarly compensated in case of tardiness.

We consider restricted instances of the $1\left\vert {}\right\vert \Sigma w_{j}U_{j}$ problem where one or more of the above three parameters is relatively small in comparison to the total input length. It is important to note that any of the parameters may be small even if the actual due dates, processing times, or weights are large numbers. For instance, all jobs can have the same processing time of $2^n$. The first result in this context is the NP-hardness result mentioned above by Karp, who showed the the 0-1~Knapsack problem reduces to the special case of $1\left\vert {}\right\vert \Sigma w_{j}U_{j}$ where all due dates are equal.
\begin{theorem}[\cite{Kar72}]
\label{theorem:d}
The $1\left\vert {}\right\vert \Sigma w_{j}U_{j}$ problem is \textnormal{NP}-hard even for the case where $nu_{d}=1$.
\end{theorem}

Theorem~\ref{theorem:d} rules out the possibility, under the P$\neq$NP hypothesis, that $1\left\vert {}\right\vert \Sigma w_{j}U_{j}$ is \emph{fixed-parameter tractable} with respect to $nu_{d}$. That is, we cannot hope for an algorithm solving the problem with a running-time of $f(nu_{d}) \cdot n^{O(1)}$ for any computable function $f()$. Our first result complements this fact. We show that a fixed-parameter algorithm for $1\left\vert {}\right\vert \Sigma w_{j}U_{j}$ is obtainable when combining $nu_{d}$ with any of the two remaining parameters $nu_{p}$ or $nu_{w}$. Moreover, we show that the problem is also fixed-parameter tractable when parameterized by $nu_{p}+nu_{w}$, positively resolving all cases of parameter combinations.
\begin{theorem}
\label{theorem:fpt}
The $1\left\vert {}\right\vert \Sigma w_{j}U_{j}$ problem is fixed-parameter tractable when parameterized by either $nu_{d}+nu_{p}$, $nu_{d}+nu_{w}$, or $nu_{p}+nu_{w}$.
\end{theorem}

We also complement Theorem~\ref{theorem:d} in another way. Obviously, the theorem rules out the possibility for a polynomial-time algorithm for the $1\left\vert {}\right\vert \Sigma w_{j}U_{j}$ problem with a constant number of due dates. Building upon the main idea of Moore~\cite{Moore1968}, we show that the problem is nevertheless solvable in time $O(n^{nu_{w}+1} \lg n)$ or $O(n^{nu_{p}+1} \lg n)$, showing that the problem is polynomial-time solvable in the case where either the number of different weights or processing times is constant.
\begin{theorem}
\label{theorem:poly}
The $1\left\vert {}\right\vert \Sigma w_{j}U_{j}$ problem is polynomial-time solvable when either $nu_{w}$ or $nu_{p}$ is constant.
\end{theorem}

\subsection{Roadmap and techniques}

The paper is organized as follows: The proof of Theorem~\ref{theorem:fpt} is split into three sections. In Section~\ref{section:dp} we provide a fixed-parameter algorithm with respect to parameter $nu_{d}+nu_{p}$. The algorithm is based on providing a mixed-integer convex programming (MICP) formulation to the problem with a fixed number of integer variables, and then applying Dadush \emph{et al.} algorithm for solving MICPs with a parameterized number of integer variables in FPT time~\cite{D11}. We then show, in Section~\ref{section:dw}, how our algorithm can be modified to handle parameter $nu_{d}+nu_{w}$. In Section~\ref{section:pw} we apply the Dadush \emph{et al.} algorithm in a different way to handle parameter $nu_{p}+nu_{w}$, by showing that a natural integer linear programming formulation of the $1\left\vert {}\right\vert \Sigma w_{j}U_{j}$ can be relaxed to an MILP with $O(nu_{p}+nu_{w})$ variables, and that a polynomial-time rounding procedure can lift solutions of this MILP back to solutions for our original problem. The proof of Theorem~\ref{theorem:poly} is given in Sections~\ref{section:w} and~\ref{section:p}. Section~\ref{section:w} provides a polynomial-time algorithm for constant values of $nu_{w}$ through a somewhat involved dynamic programming procedure. As the main ideas are similar for parameter $nu_{p}$, we only give a sketch of the algorithm for this parameter in Section~\ref{section:p}. 

\section{An FPT algorithm for parameter \boldmath{$nu_{d}+nu_{p}$}}
\label{section:dp}

In this section we prove the first part of Theorem~\ref{theorem:fpt}, and show that the $1\left\vert {}\right\vert \Sigma w_{j}U_{j}$ problem is fixed-parameter tractable with respect to parameter $k=nu_{d}+nu_{p}$. This is done by first providing an MICP formulation for the problem with $O(k)$ integer variables. The proof then follows directly from the result by Dadush \emph{et al}. ~\cite{D11} that show that the problem of solving an MICP is fixed-parameter tractable with respect to the number of integer variables.

Recall that $\mathcal{J}$ denotes our input set of jobs. We begin by partitioning $\mathcal{J}$ into $k$ classes, $\mathcal{S}_1,\ldots,\mathcal{S}_k$, such that all jobs in the same class have equal due dates and processing times. We slightly abuse notation and use $d_i$ and $p_i$ to denote the due date and processing time of all jobs in $\mathcal{S}_i$, $1 \leq i \leq k$. Moreover, we let $n_i=|\mathcal{S}_i|$ denote the number of jobs in each $\mathcal{S}_i$.

We begin with two key observations that will be useful for formulating our MICP. The first can be considered by now as folklore (see, \emph{e.g.}, \cite{Adamu2014}), while the second follows from an easy pairwise interchange argument.
\begin{lemma}
\label{lemma:dp1}%
There exist an optimal schedule where the early jobs are scheduled first in a non-decreasing order of their due-dates (i.e., according to the earliest due date (EDD) rule), while the tardy jobs are scheduled last in an arbitrary order.
\end{lemma}
\begin{lemma}
\label{lemma:dp2}%
Suppose $\pi$ is an optimal schedule which includes $y_{i}$ tardy jobs from each $\mathcal{S}_i$, $1 \leq i \leq k$. Then these tardy jobs are the $y_{i}$ least weighted jobs in $\mathcal{S}_i$.
\end{lemma}
\begin{proof}
By contradiction, consider an optimal schedule $\pi$ that does not follow the lemma statement, \emph{i.e.}, in $\pi$ there are two jobs $J_\ell, J_m \in \mathcal{S}_i$ with $w_\ell<w_m$, and $J_m$ is tardy while $J_\ell$ is not. Now, construct an alternative schedule $\pi'$ out of $\pi$ by interchanging the positions of jobs $J_\ell$ and $J_m$. Since both jobs belong to the same class $\mathcal{S}_i$, they share the same processing time and due date. Thus, $J_\ell$ is tardy in $\pi'$ while $J_m$ is not. The fact that the completion time of all other jobs remains unchanged, implies that $\Sigma w_{j}U_{j}(\pi')-\Sigma w_{j}U_{j}(\pi)=w_\ell-w_m<0$, contradicting our assumption that $\pi$ is an optimal schedule.
\end{proof}

Following Lemma \ref{lemma:dp1}, we may assume without loss of generality that the job sets $\mathcal{S}_i$ are indexed according to the EDD rule (\emph{i.e.}, that $d_{1}\leq d_{2}\leq \ldots\leq d_{k}$). Moreover, adhering to Lemma~\ref{lemma:dp2}, we order the jobs in each set $\mathcal{S}_i$ in non-decreasing order of weights. Thus, letting $J_{i,j}$ denote the $j$'th job in $\mathcal{S}_i$, we have $w_{i,1}\leq w_{i,2}\leq \ldots\leq w_{i,n_i}$, for all $i=1,\ldots,k$.

We are now ready to present our MICP formulation. Let $x_{i}$ and $y_{i}$ be two integer non-negative variables, respectively, representing the number of early jobs and tardy jobs in job set $\mathcal{S}_i$, for each $i=1,\ldots,k$. To make sure that the number of early and tardy jobs in each set is equal to
the total number of jobs in the set, we include the following constraint for each $i=1,\ldots,k$:
\begin{equation}
\label{eqn:dp1}%
x_{i}+y_{i}=n_{i}.
\end{equation}%
Moreover, to make sure that indeed it is feasible to have $x_i$ early jobs in $\mathcal{S}_i$, we add the following constraints as well for each $i=1,\ldots,k$:%
\begin{equation}
\label{eqn:dp2}
\sum_{j=1}^{i}p_{j}x_{j}\leq d_{i}.
\end{equation}%

Note that any feasible solution with respect to the $k$ constraints of type (\ref{eqn:dp1}) and the $k$ constraints of type (\ref{eqn:dp2}) corresponds to a feasible schedule where indeed $x_i$ and $y_i$ jobs are scheduled prior and after $d_i$ in each $\mathcal{S}_i$. Now, let $z_{i}$ be a new non-negative (and non-integer) variable representing the total weight of all tardy jobs in $\mathcal{S}_i$. By Lemma \ref{lemma:dp2} and the ordering of jobs within each set $\mathcal{S}_i$, we have that $z_{i}=\sum_{j=1}^{y_i}w_{i,j}$. Therefore, the objective of our MICP, corresponding to the total weight $Z$ of the set of tardy jobs is
\begin{equation}
\label{eqn:obj1}
Z=\sum_{i=1}^{k}z_{i}= \sum_{i=1}^{k} \sum_{j=1}^{y_i}w_{i,j}.
\end{equation}%

Accordingly, our problem can be solved by solving formulation $\Pi$ where we need to minimize the objective in (\ref{eqn:obj1}) subject to the set of linear constraints in (\ref{eqn:dp1}) and (\ref{eqn:dp2}).

\begin{lemma}
\label{lemma:dp3}
The objective in (\ref{eqn:obj1}) is a convex function.
\end{lemma}

\begin{proof}
As the sum of convex functions is a convex function, to prove that (\ref{eqn:obj1}) is indeed a convex function, it is enough to show that each $z_i=z_i(y_i)=\sum_{j=1}^{y_i}w_{i,j}$ function ($i=1,...,k$) is a convex function. To do so, it is enough to show that $f(y_i)=z_i(y_i+1)+z_i(y_i-1)-2z_i(y_i) \geq 0$. It is easy to show that $f(y_i)=w_{i,y_{i+1}}-w_{i,y_i}$, which is indeed a non-negative value due to the fact that we order the jobs in each $\mathcal{S}_i$ such that $w_{i,1}\leq w_{i,2}\leq \ldots\leq w_{i,n_i}$, for all $i=1,\ldots,k$.
\end{proof}

Since formulation $\Pi$ includes only $O(k)$ integer variables, a set of linear constraints and a convex objective function, we can use Dadush's \emph{et al.} algorithm for MICPs with a parameterized number of integer variables~\cite{D11}, to solve our problem in FPT time with respect to $k$. Thus, we complete the proof of the first part of Theorem~\ref{theorem:fpt}.

We next show that there is an equivalent mixed integer linear programming (MILP) formulation to formulation $\Pi$. This will enable us to use an alternative FPT algorithm for our problem by applying Lenstra's classical FPT algorithm for solving MILPs with a parameterized number of integer variables~\cite{Len83}.  Since the set of constraints in (\ref{eqn:dp1}) and (\ref{eqn:dp2}) are linear, we need only to deal with the non-linearity of the objective function. To overcome this difficulty we replace the objective function in (\ref{eqn:obj1}) by a linear objective function of minimizing $Z=\sum_{i=1}^{k}z_{i}$. Then, correspond to each set $\mathcal{S}_i$ ($i=1,...,k$), we include a set of $n_i$ linear constraints. These constraints enforce stricter and stricter lower bounds on the values of $z_i$, where the largest of these ensures that $z_i$ gets its intended value. Accordingly, for each $i=1,\ldots,k$, we add the following set of $n_i$ constraints:
\begin{equation}
\label{eqn:dp3}
z_{i}\geq (y_{i}-j+1)w_{i,j}+\sum_{\ell=1}^{j-1}w_{i,\ell}\text{ for all }j=1,\ldots,n_{i}.
\end{equation}
Consider now the problem $\Pi'$ where we wish to minimize the linear function $\sum_{i=1}^k z_i$ subject to the constraints of type (\ref{eqn:dp1}),(\ref{eqn:dp2}), and (\ref{eqn:dp3}). The following lemma shows that indeed an optimal solution for $\Pi'$ satisfies our intended meaning for each variable $z_{i}$, which further implies that $\Pi$ and $\Pi'$ are equivalent.

\begin{lemma}
\label{lemma:dp3}
In an optimal solution for $\Pi$ we have $z_{i}=\sum_{j=1}^{y_{i}}w_{i,j}$ for each $i=1,\ldots,k$.
\end{lemma}
\begin{proof}
Consider first the set of $n_i$ constraints of type (\ref{eqn:dp3}) for some fixed $i \in \{1,\ldots,k\}$. Due to the fact that the jobs in $\mathcal{S}_i$ are ordered in non-decreasing order of weights, we know that for each $j \in \{1,\ldots, y_i\}$ we have
\begin{equation*}
(y_{i}-j+1)w_{ij}+\sum_{\ell=1}^{j-1}w_{i,\ell}\leq (y_{i}-j)w_{i,j+1}+\sum_{\ell=1}^{j}w_{i,\ell},
\end{equation*}
and for each $j \in \{y_i+1,\ldots,n_i\}$ we have
\begin{equation*}
(y_{i}-j+1)w_{ij}+\sum_{\ell=1}^{j-1}w_{i,\ell}\geq (y_{i}-j)w_{i,j+1}+\sum_{\ell=1}^{j}w_{i,\ell}.
\end{equation*}%

The fact that both inequalities above hold leads to the conclusion that the maximal value of the right side of all $n_i$ constraints of type (\ref{eqn:dp3}) corresponding to $\mathcal{S}_i$ is when either $j=y_{i}$ or $j=y_{i}+1$. Note that in both cases this term is equal to $\sum_{j=1}^{y_{i}}w_{i,j}$, and so we have $
z_i \geq \sum_{j=1}^{y_{i}}w_{i,j}$. Since $z_i$ does not appear in any other constraint, and since the objective in $\Pi$ is to minimize $\sum_i z_i$, this inequality holds in equality in any optimal solution for $\Pi$.
\end{proof}

The fact that $\Pi'$ is a MILP which is equivalent to $\Pi$ implies that we can also use Lenstra's FPT algorithm for MILPs with a parameterized number of integer variables~\cite{Len83}, as a tool to solve our problem in FPT time.

\section{An FPT algorithm for parameter \boldmath{$nu_{d}+nu_{w}$}}
\label{section:dw}

In this section we prove the second part of Theorem~\ref{theorem:fpt} by providing a fixed-parameter algorithm with respect to $k=nu_d+nu_w$. Our algorithm for this parameter is very similar to the $nu_d+nu_p$ case. Again, we show that the $1\left\vert {}\right\vert \Sigma w_{j}U_{j}$ problem can be formalized as an MICP with $O(k)$ integer variables, and apply Dadush's \emph{et al}. FPT algorithm for MICPs with a parameterized number of variables~\cite{D11}.

Partition the set of input jobs $\mathcal{J}$ into $k$ classes, $\mathcal{S}_1,\ldots,\mathcal{S}_k$, such that all jobs belonging to $\mathcal{S}_i$ have the same due date $d_{i}$ and same weight $w_{i}$. Let $n_{i}=|\mathcal{S}_i|$ for each $i=1,\ldots,k$. To formulate our MICP, we use Lemma~\ref{lemma:dp1} along with the observation below:

\begin{lemma}
\label{lemma:dw1}
Consider an optimal schedule that follows the structure in Lemma~\ref{lemma:dp1} and includes $y_{i}$ tardy jobs from each~$\mathcal{S}_i$, $1 \leq i \leq k$. Then there is an optimal schedule that follows Lemma~\ref{lemma:dp1} such that for each $i$, $1 \leq i \leq k$, the set of tardy jobs includes the $y_{i}$ jobs in $\mathcal{S}_i$ with the largest processing time.
\end{lemma}

\begin{proof}
Consider an optimal schedule $\pi$ that follows the structure in Lemma~\ref{lemma:dp1}, and does not follow the structure in the above lemma. That is, $\pi$ includes at least one pair of jobs $J_\ell, J_m \in \mathcal{S}_i$ with $p_\ell<p_m$, and $J_m$ is early in $\pi$ while $J_\ell$ is not. Now, construct an alternative schedule $\pi'$ from $\pi$ by interchanging the positions of jobs $J_\ell$ and $J_m$. The fact that $p_\ell<p_m$ implies that $J_\ell$ is completed in $\pi'$ prior to the completion time of $J_m$ in $\pi$, and that the completion time of $J_m$ in $\pi'$ is equal to the completion time of $J_\ell$ in $\pi$. As both jobs share the same due date and weight, and the completion time of each of the other jobs in $\pi'$ is not later then its completion time in $\pi$, we can conclude that schedule $\pi'$ is optimal as well. By performing a similar pairwise interchange on any such pair of jobs, we end up with a schedule as in the statement of the lemma.
\end{proof}

Following Lemma~\ref{lemma:dp1}, we may assume without loss of generality that $d_{1}\leq d_{2}\leq \ldots\leq d_{k}$. Following Lemma~\ref{lemma:dw1}, we sort each $\mathcal{S}_i$ in non-decreasing order of processing times, so that if~$J_{i,j}$ is the $j$'th job in $\mathcal{S}_i$, then $p_{i,1}\leq p_{i,2}\leq \ldots\leq p_{i,n_{i}}$, for each $i=1,\ldots,k$.

As in Section~\ref{section:dp}, let $x_{i}$ and $y_{i}$, respectively, be non-negative integer variables that represent the number of early jobs and tardy jobs in $\mathcal{S}_i$, for $i=1,\ldots,k$. The objective of our MICP is given by

\begin{equation}
\label{eqn:dw}%
Z =\sum_{i=1}^{k}w_{i}y_{i}.
\end{equation}

Again, we add the $k$ constraints of type (\ref{eqn:dp1}) to ensure that the pairs $x_i,y_i$ sum up appropriately. Moreover, we add $k$ new non-negative non-integer variables $z_1,\ldots,z_k$, where $z_i$ represents the total processing time of the early jobs in $\mathcal{S}_i$. To ensure that indeed it is feasible to have $x_i$ early jobs in each set $\mathcal{S}_i$, the following constraint is included for each $i=1,\ldots,k$:
\begin{equation}
\label{eqn:dw1}%
\sum_{j=1}^{i}z_{j}\leq d_{i}.
\end{equation}%

According to Lemma~\ref{lemma:dw1} and the ordering of each $\mathcal{S}_i$ ($i=1,...,k$), we have that

\begin{equation}
\label{eqn:dw2}
z_{i}=\sum_{j=1}^{x_{i}}p_{i,j}.
\end{equation}

Accordingly, our problem can be solved by solving formulation $\Pi$ where we need to minimize the linear objective in (\ref{eqn:dw}) subject to the set of linear constraints in (\ref{eqn:dp1}) and (\ref{eqn:dw1}), and the set of constraints in (\ref{eqn:dw2}).

\begin{lemma}
\label{lemma:dw3}
Each of the $k$ constraints in (\ref{eqn:dw2}) is a convex function.
\end{lemma}

\begin{proof}
Let $z_i=z_i(x_i)=\sum_{j=1}^{x_{i}}p_{i,j}$ ($i=1,...,k$). To prove that $z_i(x_i)$ is a convex function, it is enough to show that $f(x_i)=z_i(x_i+1)+z_i(x_i-1)-2z_i(x_i) \geq 0$. It is easy to show that $f(x_i)=p_{i,x_{i+1}}-p_{i,x_i}$, which is indeed a non-negative value due to the fact that we order the jobs in each $\mathcal{S}_i$ such that $p_{i,1}\leq p_{i,2}\leq \ldots\leq p_{i,n_{i}}$.
\end{proof}

Following Lemma~\ref{lemma:dw3}; the fact that all other functions in $\Pi$ are linear; and that $\Pi$ includes $O(k)$ integer variables leads to the the conclusion that we can use Dadush's \emph{et al} FPT algorithm for MICPs with a parameterized number of integer variables~\cite{D11} to solve our problem, completing the proof of the second part of Theorem~\ref{theorem:fpt}.

Here as well, we can provide an equivalent MILP formulation (denoted by $\Pi'$) to formulation $\Pi$. This enables us to use an alternative FPT algorithm for our problem by applying Lenstra's classical FPT algorithm for solving MILPs with a parameterized number of integer variables.  Since the objective  in (\ref{eqn:dw}) and the constraints in (\ref{eqn:dp1}) and (\ref{eqn:dw1}) are all linear, we need only to deal with the non-linearity of the set of constraints in (\ref{eqn:dw2}). To overcome this difficulty, for each $i=1,...,k$ we replace the convex constraint in (\ref{eqn:dw2}) by the following set of linear constraints:
\begin{equation}
\label{eqn:dw3}%
z_{i}\geq (x_{i}-j+1)p_{i,j}+\sum_{\ell=1}^{j-1}p_{i,\ell}\text{ for all } j=1,\ldots,n_{i}.
\end{equation}%

Accordingly, in problem $\Pi'$ our goal is to minimize the objective  in (\ref{eqn:dw}) subject to the set of linear constraints in (\ref{eqn:dp1}), (\ref{eqn:dw1}) and (\ref{eqn:dw3}). Therefore, $\Pi'$ is an MILP formulation. The following lemma complete the proof that $\Pi'$ is indeed equivalent to $\Pi$ (the proof of the following lemma is very similar to that of Lemma~\ref{lemma:dp3}, and is therefore left to the reader).

\begin{lemma}
\label{lemma:dw4}
In an optimal solution for $\Pi$ we have $z_{i}=\sum_{j=1}^{x_i}p_{i,j}$, for each $i=1,\ldots,k$.
\end{lemma}

The fact that $\Pi'$ is an MILP which is equivalent to $\Pi$ and includes only $O(k)$ integer variables, implies that we can use also Lenstra's FPT algorithm for MILPs with a parameterized number of integer variables~\cite{Len83}, as a tool to solve our problem in FPT time.

\section{An FPT algorithm for parameter \boldmath{$nu_{p}+nu_{w}$}}
\label{section:pw}

Next we consider the third and final part of Theorem~\ref{theorem:fpt} concerning parameter $k=nu_p+nu_w$. Our algorithm for this case will be slightly different from the previous two cases. We again use the powerful tool of MILPs, but this time in an alternative manner. First we will formulate the problem as an integer linear program (ILP) that has $O(n+k)$ integer variables. Then we by relaxing the a subset of the integer variables we obtain a MILP relaxation of the original formulation with $O(k)$ integer variables. Finally, we prove that any optimal solution to the MILP relaxation can be rounded in linear time to a feasible solution for the ILP without any lose in the objective value.

Similar to previous sections, we begin by partitioning $\mathcal{J}$ into $k$ subsets, $S_{1},\ldots,S_{k}$, such that all jobs belonging to set $\mathcal{S}_i$ have the same processing time $p_i$, and weight $w_i$. Moreover, let $n_{i}=|\mathcal{S}_i|$ for $i=1,\ldots,k$ as usual. Our ILP formulation is based on exploiting the following lemma:
\begin{lemma}
\label{lemma:pw1}%
If $x_{i}$ is the optimal number of early jobs in $\mathcal{S}_i$ then there exists an optimal solution in which the $x_{i}$ jobs with the latest due date in $\mathcal{S}_i$ are early.
\end{lemma}

\begin{proof}
Consider an optimal schedule $\pi$ that includes at least one pair of jobs $J_\ell, J_m \in \mathcal{S}_i$ with $d_\ell<d_m$, and $J_m$ is tardy in $\pi$ while $J_\ell$ is not. Note that in $\pi$, job $J_\ell$ is scheduled prior to $J_m$, as otherwise $J_\ell$ is tardy as well.  Now, construct an alternative schedule $\pi'$ from $\pi$ by interchanging the positions of jobs $J_\ell$ and $J_m$. The fact that $p_\ell=p_m$ implies that $J_\ell$ is completed in $\pi'$ at the completion time of $J_m$ in $\pi$, and that the completion time of $J_m$ in $\pi'$ is equal to the completion time of $J_\ell$ in $\pi$. Thus, based on the fact that $d_\ell<d_m$, we can conclude that $J_m$ is early while $J_\ell$ is late in $\pi'$. The fact that both jobs share the same weight, and that the completion time of each of the other jobs remains unchanged leads to the conclusion that $\pi'$ is optimal as well. By preforming a similar pairwise interchange on any such pair of jobs, we end up with a schedule that satisfies the lemma above.
\end{proof}

Let $d_1,\ldots,d_{n_d}$ be set of due dates in our input job set $\mathcal{J}$, and assume without loss of generality that $d_{1}\leq d_{2}\leq \ldots\leq d_{n_d}$. Moreover, let $\delta_{i,j}$ be the number of jobs in $\mathcal{S}_i$ having a due date of $d_{j}$, for $i=1,\ldots,k$ and $j=1,\ldots,n_d$. We present an ILP formulation for the $1\left\vert {}\right\vert \Sigma w_{j}U_{j}$ problem, denoted by $\Pi_1$, which has $O(n+k)$ integer variables. For this, define first a set of $k$ non-negative integer variables $y_1,\ldots,y_k$ representing the number of tardy jobs in each $\mathcal{S}_i$. The objective function of $\Pi_1$ is to minimize the total weighted number of tardy jobs given by $Z=\sum^k_{i=1} w_iy_i$.

For $i\in \{1,\ldots,k\}$ and $j\in \{1,\ldots,n_d\}$, let $x_{i,j}$ be an integer variable representing the number of early jobs in $\mathcal{S}_i$ that have a due date of $d_{j}$. By definition, we have that
\begin{equation}
\label{eqn:pw1}%
x_{i,j}\leq \delta_{i,j} \text{ for all }i\in \{1,\ldots,k\} \text{ and all } j\in \{1,\ldots,n_d\},
\end{equation}%
and that%
\begin{equation}
\label{eqn:pw2}%
n_{i}-\sum_{j=1}^{n_d}x_{i,j}=y_{i}  \text{ for all }i\in \{1,\ldots,k\}.
\end{equation}%
Finally, to make sure that each early job is completed not later then its corresponding due date, we include the following set of constraints:
\begin{equation}
\label{eqn:pw3}%
\sum_{i=1}^{k} \sum_{j=1}^{\ell} p_{i}x_{i,j}\leq d_{\ell} \text{ for all } \ell \in \{1,\ldots,n_d\}.
\end{equation}%

Thus, $\Pi_1$ is the problem of minimizing $Z=\sum_i w_iy_i$ subject to all constraints of type (\ref{eqn:pw1}),(\ref{eqn:pw2}), and (\ref{eqn:pw3}). Note that $\Pi_1$ is indeed an ILP formulation for the $1\left\vert {}\right\vert \Sigma w_{j}U_{j}$ problem. However it has too many integer variables to apply either Dadush's \emph{et al}. or Lenstra's algorithm directly. To circumvent this, we define an MILP $\Pi_2$ where we relax that constraint that all the $x_{i,j}$'s must be integer, and only require that they have to be non-negative. In this way, $\Pi_2$ is an MILP relaxation of $\Pi_1$ with $O(k)$ integer variables, and we can compute an optimal solution for $\Pi_2$ in FPT time with respect to $k$ using either Dadush's \emph{et al}. or Lenstra's algorithm.

Let $(x^*,y^*)$, where $x^*=(x^*_{i,j})$ and $y^*=(y^*_{i})$ for $i\in \{1,\ldots,k\}$ and $j\in \{1,\ldots,\ell\}$, be an optimal solution for $\Pi_2$, and let $x^*_i=\sum_{j=1}^{n_d}x^*_{i,j}$ for $i\in \{1,\ldots,k\}$. Note that $x^*_i$ is an integer value for $i\in \{1,\ldots,k\}$ due to (\ref{eqn:pw2}) and the fact that both $n_i$ and $y_i$ are integer values.
Now, if $(x^*,y^*)$ is a feasible solution for $\Pi_1$ (i.e., all $x^*_{i,j}$ are assigned integer values), then $(x^*,y^*)$ is also an optimal solution for $\Pi_1$. Otherwise, the value of some of the $x^*_{i,j}$ variables is not integer. In this case, we use the following rounding procedure to obtain an alternative optimal solution $(\tilde{x},y^*)$ for $\Pi_2$ which will be feasible also for $\Pi_1$:

\begin{quote}
\emph{Rounding Procedure}: For each $i=1,\ldots,k$, let $r_i$ be the integer satisfying
$$
\sum_{j=r_{i}+1}^{n_d}\delta_{i,j} \leq x^*_i < \sum_{j=r_{i}}^{n_d}\delta_{i,j}.
$$
Define
$$
\tilde{x}_{i,j} =
\begin{cases}
0 &\text{for } j=1,\ldots,r_i-1,\\
x^*_i-\sum_{j=r_{i}+1}^{n_d} \delta_{i,j} &\text{for } j=r_i,\\
\delta_{i,j} &\text{for } j=r_i+1,\ldots,n_d.
\end{cases}
$$
\end{quote}

\begin{lemma}
\label{lol2}
$(\tilde{x},y^*)$ is an optimal solution for $\Pi_1$.
\end{lemma}

\begin{proof}
Note that $(\tilde{x},y^*)$ and $(x^*,y^*)$ have the same objective value in $\Pi_2$. Thus, to show that $(\tilde{x},y^*)$ is an optimal solution for $\Pi_1$, it is enough to show that it is feasible in $\Pi_1$.

First, observe that the values of all variables in $(\tilde{x},y^*)$ is integer, since the value of each $\tilde{x}^*_{i,j}$ is defined by the integer variable $x^*_i$ and the integers $\delta_{i,j}$ for $j=1,\ldots,n_d$. Thus, it remains to argue that $(\tilde{x},y^*)$ satisfies all constraints of type (\ref{eqn:pw1}), (\ref{eqn:pw2}), and (\ref{eqn:pw3}). Note that by definition of $\tilde{x}$, we have $\tilde{x}_{i,j} \leq \delta_{i,j}$ for each $i\in \{1,...,k\}$ and $j\in \{1,...,n_d\}$, and so $(\tilde{x},y^*)$ satisfies all constraints of type (\ref{eqn:pw1}). Thus, the only interesting cases to consider are constraints of type (\ref{eqn:pw2}) and of type (\ref{eqn:pw3}).

Consider an optimal solution $(x^*,y^*)$ for $\Pi_2$. By definition of $\tilde{x}$, we have for each $i=1,\ldots,k$:
\begin{eqnarray*}
\tilde{x}_i \quad = \quad  \sum_{j=1}^{n_d}\tilde{x}_{i,j}  & = \quad  \sum_{j=1}^{r_i-1} \tilde{x}_{i,j} + \tilde{x}_{i,r_{i}} + \sum_{j=r_i+1}^{n_d} \tilde{x}_{i,j} & = \\  & 0 + (x^*_i-\sum_{j=r_i+1}^{n_d} \delta_{i,j})+\sum_{j=r_i+1}^{n_d} \delta_{i,j} & = \quad  x^*_i  \quad =\quad \sum_{j=1}^{n_d}x^*_{i,j}.
\end{eqnarray*}
Thus, $\sum_{j=1}^{n_d}\tilde{x}_{i,j} = \sum_{j=1}^{n_d}x^*_{i,j}$. Since $(x^*,y^*)$ satisfies all constraints of type (\ref{eqn:pw2}), we have
$$
n_{i}-\sum_{j=1}^{n_d} \tilde{x}_{i,j} = n_{i}-\sum_{j=1}^{n_d} x^*_{i,j} =y^*_i
$$
for each $i=1\ldots,k$, and so $(\tilde{x},y^*)$ also satisfies all constrains of type (\ref{eqn:pw2}).

Next, observe that the vector of variables $\tilde{x}$ minimizes the value of $\sum_{\ell=1}^{j}x_{i,\ell}$ for each $j=1,...,n_d$, subject to the restriction that $\sum_{j=1}^{n_d} x_{i,j} = x^*_i$ and $x_{i,j} \leq \delta_{i,j}$ for all $i=1,\ldots,k$ and $j=1,\ldots,n_d$. Therefore, since $x^*$ is also subject to these restrictions, we have
$$
\sum_{i=1}^{k}\sum_{\ell=1}^{j}p_{i}\tilde{x}_{i,\ell} \leq \sum_{i=1}^{k}\sum^j_{\ell=1} p_ix^*_{i,\ell} \leq d_{j},
$$
and so $(\tilde{x},y^*)$ satisfies all constraints of type (\ref{eqn:pw3}) as well.
\end{proof} 

\section{A polynomial-time algorithm for constant \boldmath{$nu_{w}$}}
\label{section:w}

In this section we provide an $O(n^{nu_w+1}\lg n)$ time dynamic programming algorithm for the $1\left\vert {}\right\vert \Sigma w_{j}U_{j}$ problem, proving the first part of Theorem~\ref{theorem:poly}. Let $w_{1},\ldots,w_{nu_w}$ denote the set of all different weights of $\mathcal{J}$. We say that job $J_{j} \in \mathcal{J}$ is of type $i$ if its weight is $w_{i}$. Furthermore, we assume that the jobs in $\mathcal{J}=\{J_1,\ldots,J_n\}$ are ordered according to the EDD rule, and so $d_{1}\leq d_{2}\leq \ldots\leq d_{n}$.

Throughout the section we will only be concerned with schedules that satisfy the properties of Lemma~\ref{lemma:dp1}. That is, schedules where all early jobs are scheduled first in an EDD order followed by the tardy jobs which are scheduled arbitrarily. Consider such a schedule $\pi_j: \mathcal{J}_j \to \{1,\ldots,j\}$ for the set of jobs $\mathcal{J}_j=\{J_1,\ldots,J_{j}\}$, for some $j < n$. We say that a schedule $\pi_{j+1}: \mathcal{J}_{j+1} \to \{1,\ldots,j+1\}$ is an \emph{extension} of $\pi_j$ if $\pi_{j+1}(J_i) = \pi_j(J_i)$ for every early job $J_i$ in $\pi_j$. 
For a pair of schedules $\pi_j^1,\pi_j^2 : \mathcal{J}_j \to \{1,\ldots,j\}$, we say that $\pi_j^1$ \emph{dominates} $\pi_j^2$ if an optimal extension of $\pi_j^1$ has an objective value not greater than an optimal extension of $\pi_j^2$. We have the following elimination property:
\begin{lemma}
\label{lemma:w1}%
Let $\pi_j^1$ and $\pi_j^2$ be two schedules for $\mathcal{J}_j$, both with $e_i$ early jobs of type $i$ for each $i=1,\ldots,nu_w$. Moreover, let $P_1$ and $P_2$ denote the total processing time of the $e=\sum^{nu_w}_{i=1} e_i$ early jobs of $\pi_j^1$ and $\pi_j^2$, respectively. If $P_1\leq P_2$, then $\pi_j^1$ dominates $\pi_j^2$.
\end{lemma}

\begin{proof}
First note that the two schedules have the same partial objective value of $\Sigma_{i=1}^{nu_w} w_{i}(n_{i,j}-e_{i})$, where $n_{i,j}$ is the number of jobs of type $i$ in job set $\mathcal{J}_j=\{J_{1},...,J_{j}\}$. Consider now an optimal extension of $\pi_j^2$ to a schedule $\pi^2_{j+1}$. Then $\pi^2_{j+1}(J_{j+1}) \geq e+1$ by the structure of Lemma~\ref{lemma:dp1}. Let $\pi^1_{j+1}$ be an extension of $\pi_j^1$ obtained by setting $\pi^1(J_{j+1}) = e+1$. The fact that $P_1\leq P_2$ implies that if $J_{j+1}$ is early in $\pi^2_{j+1}$, it is also early in $\pi^1_{j+1}$. Therefore, $\pi^1_{j+1}$ has an objective value not greater than that of $\pi^2_{j+1}$.
\end{proof}

Let us say that a schedule $\pi$ is of \emph{category} $(e_1,\ldots,e_{nu_w}) \in \{1,\ldots,n\}^{nu_w}$, if there are exactly $e_i$ early jobs of type $i$ in $\pi$, for each $i=1,\ldots,nu_w$. Based on Lemma \ref{lemma:w1}, we next present a dynamic programming algorithm that constructs the set of all schedules that dominate all other schedules in their category, for each category $(e_1,\ldots,e_{nu_w})$ and each subset of jobs $\mathcal{J}_j$. To do so, let $P_{j}(e_1,\ldots,e_{nu_w})$ denote the minimum total processing time of the early jobs among all schedules for $\mathcal{J}_j$ of category $(e_1,\ldots,e_{nu_w})$.



To compute the values $P_{j}(e_1,\ldots,e_{nu_w})$, we maintain data structures that allow us some bookkeeping. For each job index $j$ and category $(e_1,\ldots,e_{nu_w})$, we maintain a heap $H_{i,j}(e_1,\ldots,e_{nu_w})$ that contains the processing times of $e_i$ early jobs of type $i$. These processing times will be kept as small as possible throughout the computation, so that they correspond to the minimum total processing times of type $i$ early jobs in any schedule for $\mathcal{J}_j$ of category $(e_1,\ldots,e_{nu_w})$. A standard implementation of $H_{i,j}(e_1,\ldots,e_{nu_w})$ allows us to insert or remove an element from $H_{i,j}(e_1,\ldots,e_{nu_w})$ in $O(\lg n)$ time, as well as obtain the maximum value $h_{i,j}(e_1,\ldots,e_{nu_w})$ of $H_{i,j}(e_1,\ldots,e_{nu_w})$ in $O(1)$ time.

Our dynamic program computes the values $P_j(e_1,\ldots,e_{nu_w})$ in increasing $j$ and $e_1,\ldots,e_{nu_w}$, using the maximum values of the heaps computed in previous steps. Assume that job $J_j$ is of type~$i$.
We consider the following two cases to compute $P_j(e_{1}\ldots,e_{nu_w}):$\\
\begin{itemize}
\item \underline{$P_{j-1}(e_{1},\dots,e_{i}-1,\dots,e_{nu_w})+p_{j}> d_{j}$}: In this case we can only use $J_j$ to replace the early job of type $i$ with the maximum processing time in a schedule of the same category for $\mathcal{J}_{j-1}$. Accordingly, we set
$$
P_j(e_{1}\ldots,e_{nu_w})=P_{j-1}(e_{1},\ldots,e_{nu_w})+\min\{0,p_{j}-h_{i,j-1}(e_{1},\ldots,e_{nu_w})\}.
$$
We construct the set of $k$ heaps corresponding to this value in the natural manner. We first set $H_{i_0,j}(e_{1}\ldots,e_{nu_w}) = H_{i_0,j-1}(e_{1}\ldots,e_{nu_w})$ for all $i_0=1,\ldots,k$. Then, if $p_j < h_{i,j}(e_{1}\ldots,e_{nu_w})$, we remove $h_{i,j}(e_{1}\ldots,e_{nu_w})$ from $H_{i,j}(e_{1}\ldots,e_{nu_w})$, and add $p_j$ instead.
\\

\item \underline{$P_{j-1}(e_{1},\dots,e_{i}-1,\dots,e_{nu_w})+p_{j}\leq d_{j}$}: In this case, we can also safely add $J_j$ to set of early jobs of the schedule corresponding to $P_{j-1}(e_{1},\dots,e_{i}-1,\dots,e_{nu_w})$, placing him last among all early jobs. Therefore, we have
$$
P_j(e_{1}\ldots,e_{nu_w})=\min
\begin{cases}
P_{j-1}(e_{1},\ldots,e_{nu_w})+\min\{0,p_{j}-h_{i,j-1}(e_{1},\ldots,e_{nu_w})\}, \\
P_{j-1}(e_{1},\ldots,e_{i}-1,\ldots,e_{nu_w})+p_{j}.
\end{cases}
$$
In the first case of this recursion, we construct the set of $k$ heaps as above. In the second case, we first set $H_{i_0,j}(e_1,\ldots,e_{nu_w}) = H_{i_0,j}(e_{1},\ldots,e_{i}-1,\ldots,e_{nu_w})$ for all $i_0 =1,\ldots,k$. We then add $p_j$ to $H_{i,j}(e_1,\ldots,e_{nu_w})$.
\end{itemize}

We implement the above recursion and update the heaps accordingly for all $j\in \{0,\ldots,n\}$ and $e_i\in \{0,\ldots,n_{i,j}\}$, where $%
n_{i,j}$ is the number of jobs of type $i$ in $\mathcal{J}_j$. The base cases of the recursion are given by
$$
P_0(e_{1}\ldots,e_{nu_w})=
\begin{cases}
0 & :  \quad e_{1}=e_{2}=\cdots=e_{nu_w}=0, \\
\infty & :  \quad \text{otherwise}.
\end{cases}
$$
The heaps are initialized by $H_{i,0}(0,\ldots,0)=\emptyset$ for all $i\in \{1,\ldots,nu_w\}$. At the end of the dynamic programming implementation, the optimal solution value is

$$
Z^* = \min \left\{\sum_{i=1}^{nu_w} w_i(n_i-e_i) \quad : \quad P_{n}(e_1,e_2,\ldots,e_{nu_w})<\infty \right\}.
$$

This completes the description of our algorithm. The next three lemmas prove its correctness. In the first we show that each entry $P_j(e_1,\ldots,e_{nu_w}) \neq \infty$ corresponds to an actual schedule of $\mathcal{J}_j$ of category $(e_1,\ldots,e_{nu_w})$, and in the second we show that this schedule has minimum total processing time of early jobs among all schedules for $\mathcal{J}_j$ of the same category. The last lemma shows that if the entry equals $\infty$, then there is no schedule for $\mathcal{J}_j$ within that category.

\begin{lemma}
\label{lemma:w2}
Each entry $P_j(e_1,\ldots,e_{nu_w}) \neq \infty$ computed by the algorithm corresponds to a schedule for $\mathcal{J}_j$ of category $(e_1,\ldots,e_{nu_w})$ where the total processing time of early jobs equals $P_j(e_1,\ldots,e_{nu_w})$. Moreover, for each $i_0 \in \{1,\ldots,nu_w\}$, the values in $H_{i_0,j}(e_1,\ldots,e_{nu_w})$ correspond to the processing times of the $e_{i_0}$ early jobs of type $i_0$ in this schedule.
\end{lemma}
\begin{proof}
The proof is by induction on $j$. The base case of $j=0$ is immediate, so assume that $j > 0$, and that the lemma holds for $j-1$. Let $(e_1,\ldots,e_{nu_w})$ be some category, and let $i$ be the type of job $J_j$. Furthermore, let $P_j$, $P^1_{j-1}$ and $P^2_{j-1}$ be shorthand notation for $P_j(e_1,\dots,e_{nu_w})$, $P_{j-1}(e_1,\dots,e_i-1,\dots,e_{nu_w})$ and $P_{j-1}(e_1,\dots,e_{nu_w})$, respectively. Then our inductive hypothesis holds for both $P^1_{j-1}$ and $P^2_{j-1}$, and so let $\pi^1_{j-1}$ and $\pi^2_{j-1}$ respectively denote the two corresponding schedules promised by this hypothesis. We show how to construct a schedule $\pi_j$ corresponding to $P_j$ from both $\pi^1_{j-1}$ and $\pi^2_{j-1}$. Let $J \in \mathcal{S}_i \cap \mathcal{J}_{j-1}$ denote the job with the maximum processing time among all early jobs of type $i$ in $\pi^1_{j-1}$, and let $p$ and $d$ respectively denote its processing time and due date. By induction, we know that $p = h_{i,j-1}(e_{1},\dots,e_i-1,\dots,e_{nu_w})$, and since $\mathcal{J}_j$ is ordered according to the EDD rule, we also know that $d \leq d_j$. We now consider both cases of the recursion.

Suppose $P^1_{j-1} + p_j > d_j$. If $p \leq p_j$, then $P_j=P^2_{j-1}$ by the above recursion, and so in this case we construct $\pi_j$ by setting $\pi_j(J_\ell)=\pi^2_{j-1}(J_\ell)$ for all $\ell=1,\ldots,j-1$, and $\pi_j(J_j)=j$. Clearly, $\pi_j$ is of category $(e_1,\ldots,e_{nu_w})$, the total processing time of early jobs in $\pi_j$ equals $P_j$, and each heap $H_{i_0,j}(e_1,\ldots,e_{nu_w})$ contains the correct values corresponding to $\pi_j$. If $p > p_j$, then $P_j=P^2_{j-1}+p-p_j$. In this case, $\pi_j$ schedules all early jobs in $\pi^2_{j-1}$ apart from $J$ first, maintaining their order in $\pi^2_{j-1}$, then it schedules $J_j$, and then all remaining jobs (the tardy jobs in $\pi^2_{j-1}$ and $J$) are scheduled in an arbitrary order. As $p_j \leq p$ and $d_j \geq d$, we know that $J_j$ is early in $\pi_j$, and so $\pi_j$ is also of category $(e_1,\ldots,e_{nu_w})$. Thus, the total processing time of early jobs in $\pi_j$ equals $P_j$, and each heap $H_{i_0,j}(e_1,\ldots,e_{nu_w})$ contains the correct values corresponding to $\pi_j$ in this case as well.

Suppose $P^1_{j-1} + p_j \leq d_j$. If $P_j=P^2_{j-1}+\min\{0,p_j-p\}$, then we construct $\pi_j$ as above. Otherwise, $P_j=P^1_{j-1}+p_j$. In this case we construct $\pi_j$ from $\pi^1_{j-1}$ by scheduling $J_j$ immediately after all early jobs in $\pi^1_{j-1}$, followed by all tardy jobs in $\pi^1_{j-1}$. Clearly all early jobs in $\pi^1_{j-1}$ are also early in $\pi_j$, and as $P^1_{j-1} + p_j \leq d_j$, so is $J_j$. Thus, $\pi_j$ satisfies the requirement of the lemma.
\end{proof}

\begin{lemma}
\label{lemma:w3}
The schedule corresponding to each entry $P_j(e_1,\ldots,e_{nu_w}) \neq \infty$ in Lemma~\ref{lemma:w2} dominates all schedules for $\mathcal{J}_j$ in its category.
\end{lemma}

\begin{proof}
According to Lemma~\ref{lemma:w1}, it is enough to show that the schedule corresponding to each entry $P_j(e_1,\ldots,e_{nu_w}) \neq \infty$ has minimum total processing time of early jobs among all schedules for $\mathcal{J}_j$ of the same category. We prove this by induction on $j$. The base case $j=0$ is trivial, so we assume that $j >0$ and that the lemma holds for $j-1$. Let $(e_1,\ldots,e_{nu_w})$ be some category, and let $i$ be the type of job $J_j$. Furthermore, let $P_j$, $P^1_{j-1}$ and $P^2_{j-1}$ be shorthand notation for $P_j(e_1,\dots,e_{nu_w})$, $P_{j-1}(e_1,\dots,e_i-1,\dots,e_{nu_w})$ and $P_{j-1}(e_1,\dots,e_{nu_w})$, respectively, and $\pi^*_j$ denote a schedule for $\mathcal{J}_j$ of category $(e_1,\ldots,e_{nu_w})$ with total processing time of early jobs $P^*_j$ which is minimal among all schedules for $\mathcal{J}_j$ in its category. To complete the proof we argue that $P_j \leq P^*_j$.

Let $\pi^*_{j-1}$ be the schedule for $\mathcal{J}_{j-1}$ obtained by removing $J_j$ from $\pi^*_j$ and maintaining the order among all remaining jobs, and let $P^*_{j-1}$ denote the total processing time of early jobs in this schedule. Suppose first that $J_j$ is tardy in $\pi^*_j$. Then $\pi^*_{j-1}$ is of category $(e_1,\dots,e_{nu_w})$, and so $P^2_{j-1} \leq P^*_{j-1}=P^*_j$ by our inductive hypothesis on $P^2_{j-1}$. Since $P_j \leq P^2_{j-1}$ holds in the above recursion, we have $P_j \leq P^*_j$. If $J_j$ is early in $\pi^*_j$, then $P^*_{j-1}$ is of category $(e_1,\dots,e_i-1,\dots,e_{nu_w})$, and so $P^1_{j-1} \leq P^*_{j-1}$. By the recursive formula we have $P_j \leq P^1_{j-1} + p_j$, and so $P_j \leq P^1_{j-1} + p_j \leq P^*_{j-1} + p_j = P^*_j$.
\end{proof}

\begin{lemma}
\label{lemma:w43}
If $P_j(e_1,\ldots,e_{nu_w}) = \infty$ then there is no schedule of category $(e_1,\ldots,e_{nu_w})$ for the job set $\mathcal{J}_j$.
\end{lemma}
\begin{proof}
Suppose that the lemma is false. Let $P_j(e_1,\ldots,e_{nu_w})$ be an entry which is a counter example with minimum $j$, and let $\pi^*_j$ be a schedule for $\mathcal{J}_j$ of category $(e_1,\ldots,e_{nu_w})$. Clearly $j >0$, and by the minimality of $j$ there is no schedule for $\mathcal{J}_{j-1}$ of the same category. Thus, $J_j$ must be early in $\pi^*_j$. Let $\pi^*_{j-1}$ be the schedule for $\mathcal{J}_{j-1}$ that is obtained from $\pi^*_j$ by omitting $J_j$. Then if $J_j$ is of type $i$, the schedule $\pi^*_{j-1}$ is of category $(e_1,\ldots,e_i-1,\dots,e_{nu_w})$. By minimality of $j$, we have $P_{j-1}(e_1,\ldots,e_i-1,\dots,e_{nu_w}) \neq \infty$, and so by Lemma~\ref{lemma:w2} and~\ref{lemma:w3} our algorithm computes a schedule $\pi_{j-1}$ for $\mathcal{J}_{j-1}$ of category $(e_1,\ldots,e_i-1,\dots,e_{nu_w})$ which dominates $\pi^*_{j-1}$. But this is a contradiction since $\pi^*_{j-1}$ can be extended to a schedule with a better objective function than any extension of $\pi_{j-1}$.
\end{proof}

Note that running time of our algorithm is dominated by the dynamic program implementation. There are $O(n^{nu_w+1})$ different $P_j(e_1,e_2,\ldots,e_{nu_w})$ entries to compute. Each entry requires $O(\lg n)$ time, assuming we reuse the heaps for $\mathcal{J}_{j-1}$ of the same category. Thus, the total running time of our algorithm can be bounded by $O(n^{nu_w+1}\lg n)$, and the first part of Theorem~\ref{theorem:poly} holds.

\section{A polynomial-time algorithm for constant \boldmath{$nu_{p}$}}
\label{section:p}

In this section we show how to modify the algorithm of Section~\ref{section:w} to handle parameter $nu_p$. Due to space constraints, we only give a brief sketch.
In the case of parameter $nu_p$, we have $nu_p$ types of jobs, and jobs of the same type have the same processing time. As usual, we start by renumbering the jobs according to the EDD rule, and so we can assume $d_{1}\leq d_{2}\leq \ldots\leq d_{n}$. The eliminating property we use in this case is as follows (proof omitted):
\begin{lemma}
\label{lemma:p1}%
Let $\pi_j^1$ and $\pi_j^2$ be two schedules of the same category for $\mathcal{J}_j$, and let $W_1$ and~$W_2$ denote the total weight of the early jobs in $\pi_j^1$ and $\pi_j^2$, respectively. If $W_1 \geq W_2$, then $\pi_j^1$ dominates $\pi_j^2$.
\end{lemma}

Accordingly, we compute for each $j =1,\ldots,n$ and category $(e_{1},\ldots,e_{nu_p})$ the value $W_{j}(e_{1},\ldots,e_{nu_p})$, which represents the maximum total weight of the early jobs among all schedules of category $(e_1,\ldots,e_{nu_p})$ for $\mathcal{J}_j$. For this, we maintain $nu_p$ heaps for each entry, where heap $H_{i,j}(e_1,\ldots,e_{nu_p})$ contains the weights of $e_i$ early jobs of type $i$. These weights will be kept as large as possible throughout the computation, and $h_{i,j}(e_1,\ldots,e_{nu_p})$ denotes the minimum value in the heap. The two cases of the recursion are very similar to the previous section:\\
\begin{itemize}
\item \underline{$\sum_{i=1}^{nu_p} p_{i}e_{i} > d_{j}$}: In this case, we compute
$$
W_j(e_1,\ldots,e_{nu_p}) = W_{j-1}(e_1,\ldots,e_{nu_p})+\max\{0,w_j-h_{i,j-1}(e_1,\ldots,e_{nu_p})\}.
$$

\item \underline{$\sum_{i=1}^{nu_p} p_{i}e_{i} \leq d_{j}$}: In this case, we set
$$
W_j(e_1,\ldots,e_{nu_p})= \max
\begin{cases}
W_{j-1}(e_1,\ldots,e_{nu_p})+\max\{0,w_j-h_{i,j-1}(e_1,\ldots,e_{nu_p})\}, \\
W_{j-1}(e_1,\ldots,e_i-1,\ldots,e_{nu_p})+w_{j}.
\end{cases}
$$
\end{itemize}

The heaps in both cases are updated in the natural manner. Moreover, the base cases are identical to the previous section, except that we use $-\infty$ instead of $\infty$ to indicate that no schedule is possible. The proof of correctness of and runtime analysis of our algorithm can be done in the same manner as for the case of parameter $nu_w$, which leads to the proof of the second part of Theorem~\ref{theorem:poly}. 

\section{Summary and future work}

In this paper we analyze the tractability of the classical classical \emph{NP%
}-hard single machine scheduling problem with the objective of minimizing
the total weighted number of tardy jobs, when a subset of its parameters is of a limited size.
We focus on different combinations of the following three parameters. The first is the number of different due dates;
the second is the number of different weights; and the third is the number
of different processing times. We show that the problem belongs to the \emph{%
FPT} set for any combination of two parameters. We also explain why the
problem is $W[1]$-hard for the first parameter and prove that the problem
is solvable in polynomial time when either one of the two last parameters is constant.

The most important questions that were left open are
to determine whether the studied scheduling problem belongs to the \emph{FPT}
set with respect to ($i$)\ the number of different weights; and ($ii$) the
number of different processing time. Moreover, since there are only few
available results on the parameterized complexity of scheduling problems,
future work may focus on other classical hard scheduling problems, such that
of minimizing the total tardiness on a single machine, or minimizing the
makespan in a three machine flow shop system. Moreover, an extension of the studied problem for the case of arbitrary release dates and/or parallel machine setting may also be a possible direction for future research.

\bibliographystyle{plain}
\bibliography{bib}

\end{document}